%% file: paper.tex
\documentclass[11pt]{article}
\usepackage[latin1]{inputenc}
\usepackage{graphicx}
\usepackage{epsfig}
\usepackage{color}
\usepackage{listings}
\usepackage{vmargin}
\usepackage{multicol}
\usepackage{lscape}
\usepackage{moreverb}
\usepackage{verbatim}
\usepackage{amsthm, amssymb}
\usepackage{amsmath}
\usepackage{hhline}
\newtheorem{Lem}{Lemma}
\newtheorem{theorem}{Theorem}

\setpapersize{A4}
\setmarginsrb{1in}{1in}{1in}{1in}{0pt}{0mm}{0pt}{0mm}

\title{Shortest Paths in Planar Graphs with Real Lengths in $O(n\log^2n/\log\log n)$ Time}
\author{Shay Mozes
        \footnote{Department of Computer Science,
                  Brown University,
                  \texttt{shay@cs.brown.edu},
                  \texttt{http://www.cs.brown.edu/people/shay/}} \and
        Christian Wulff-Nilsen
        \footnote{Department of Computer Science,
                  University of Copenhagen,
                  \texttt{koolooz@diku.dk},
                  \texttt{http://www.diku.dk/hjemmesider/ansatte/koolooz/}}}

\date{}
\begin{document}
\setlength\footskip{36pt}

\maketitle
\begin{abstract}
Given an $n$-vertex planar directed graph with real edge lengths and with no negative cycles, we show
how to compute single-source shortest path distances in the graph in $O(n\log^2n/\log\log n)$ time with
$O(n)$ space. This is an improvement of a recent time bound of $O(n\log^2n)$ by Klein et al.
\end{abstract}

\section{Introduction}
Computing shortest paths in graphs is one of the most fundamental
problems in combinatorial optimization with a rich history. Classical shortest path algorithms are the Bellman-Ford algorithm
and Dijkstra's algorithm which both find distances from a given source vertex to all other vertices in
the graph. The Bellman-Ford algorithm works for general graphs and has running time $O(mn)$ where $m$ resp.\ $n$
is the number of edges resp.\ vertices of the graph. Dijkstra's algorithm runs in $O(m + n\log n)$ time
when implemented with Fibonacci heaps but it only works for graphs with non-negative edge lengths.

We are interested in the single-source shortest path (SSSP) problem for planar directed graphs. There is an optimal $O(n)$
time algorithm for this problem when all edge lengths are non-negative~\cite{SSSPPlanar}. For planar graphs with
arbitrary real edge lengths and with no negative cycles, Lipton, Rose, and Tarjan~\cite{LiptonRoseTarjan} gave an
$O(n^{3/2})$ time algorithm. Henzinger, Klein, Rao, and Subramanian~\cite{SSSPPlanar} obtained a (not strongly) polynomial
bound of $\tilde{O}(n^{4/3})$. Later, Fakcharoenphol and Rao~\cite{Fakcharoenphol} showed how to
solve the problem in $O(n\log^3n)$ time and $O(n\log n)$ space. Recently, Klein, Mozes, and Weimann~\cite{SSSPPlanarNeg}
presented a linear space $O(n\log^2n)$ time algorithm.

In this paper, we improve the result in~\cite{SSSPPlanarNeg} by exhibiting a linear space algorithm with
$O(n\log^2n/\log\log n)$ running time.

From the observations in~\cite{SSSPPlanarNeg}, our algorithm can be used to solve bipartite planar perfect matching,
feasible flow, and feasible circulation in planar graphs in $O(n\log^2n/\log\log n)$ time.

The organization of the paper is as follows. In Section~\ref{sec:DefsRes},
we give some definitions and some basic results, most of them related to planar graphs. Our algorithm is very similar to
that of Klein et al. so in Section~\ref{sec:KleinAlgo}, we give an overview of some of their ideas. We then show how to
improve the running time in Section~\ref{sec:ImprovedAlgo}. Finally, we make some concluding remarks in
Section~\ref{sec:ConclRem}.

\section{Definitions and Basic Results}\label{sec:DefsRes}
In the following, $G = (V,E)$ denotes an $n$-vertex planar directed graph with real edge lengths and with no
negative cycles. For vertices $u,v\in V$, we let $d_G(u,v)\in\mathbb R\cup\{\infty\}$ denote the length of a shortest path
in $G$ from $u$ to $v$. We extend this notation to subgraphs of $G$. We will assume that $G$ is triangulated such that there
is a path of finite length between each ordered pair of vertices of $G$. The new edges added have sufficiently large lengths
so that shortest path distances in $G$ will not be affected.

Given a graph $H$, we let $V_H$ and $E_H$ denote its vertex set and edge set, respectively.
For an edge $e\in E_H$, we let $l(e)$ denote the length of $e$ (we omit $H$ in the definition but this should not cause
any confusion). Let $P = u_1,\ldots,u_m$ be a path in $H$. We let $|P| = m$. For
$1\leq i\leq j\leq m$, $P[u_i,u_j]$ denotes the subpath $u_i,\ldots,u_j$. A path $P'$ is said to \emph{intersect} $P$ if
$V_P\cap V_{P'}\neq\emptyset$. If $P' = u_m,\ldots,u_{m'}$ is another path, we define
$PP' = u_1,\ldots,u_{m-1},u_m,u_{m+1},\ldots,u_{m'}$.

Define a \emph{region} $R$ (of $G$) to be the subgraph of $G$ induced by a subset of $V$. In $G$, the vertices of
$V_R$ that are adjacent to vertices in $V\setminus V_R$
are called \emph{boundary vertices} (of $R$) and the set of boundary vertices of $R$ is called the \emph{boundary} of $R$.
Vertices of $V_R$ that are not boundary vertices of $R$ are called \emph{interior vertices} (of $R$).

The cycle separator theorem of Miller~\cite{CycleSep} states that, given an $m$-vertex plane
graph, there is a Jordan curve $C$ intersecting $O(\sqrt m)$ vertices and no edges such that
between $m/3$ and $2m/3$ vertices are enclosed by $C$. Furthermore, this Jordan curve can be found in linear time.

Let $r\in(0,n)$ be a parameter. Fakcharoenphol and Rao~\cite{Fakcharoenphol} showed how to recursively apply the
cycle separator theorem such that in $O(n\log n)$ time, (a plane embedding of) $G$ is divided into $O(n/r)$ regions with some
nice properties:
\begin{enumerate}
\item each region contains at most $r$ vertices and $O(\sqrt r)$ boundary vertices,
\item no two regions share interior vertices,
\item each region has a boundary contained in $O(1)$ faces, defined by simple cycles.
\end{enumerate}
We refer to such a division as an \emph{$r$-division} of $G$. The bounded faces of a region are its
\emph{holes}. To simplify the description of our algorithm, we will refer to all vertices of the cycles containing the
boundary of a region as boundary vertices of that region. Furthermore, we will assume that for each region $R$ in an
$r$-division, $R$ is contained in the bounded region defined by one of the cycles $C$ in the boundary of $R$.
Clearly, this can always be achieved by adding a new cycle if needed. We refer to $C$ as the \emph{external face} of $R$.

For a graph $H$, a \emph{price function} is a function $p:V_H\rightarrow\mathbb R$. The
\emph{reduced cost function} induced by $p$ is the function $w_p:E_H\rightarrow\mathbb R$, defined by
\[
  w_p(u,v) = p(u) + l(u,v) - p(v).
\]
We say that $p$ is a \emph{feasible} price function for $H$ if for all $e\in E_H$, $w_p(e)\geq 0$.

It is well known that reduced cost functions preserve shortest paths, meaning that we can find shortest paths
in $H$ by finding shortest paths in $H$ with edge lengths defined by the reduced cost function $w_p$. Furthermore,
given $\phi$ and the distance in $H$ w.r.t.\ $w_p$ from a $u\in V_H$ to a $v\in V_H$, we can extract the original distance
in $H$ from $u$ to $v$ in constant time~\cite{SSSPPlanarNeg}.

Observe that if $p$ is feasible, Dijkstra's algorithm can be applied to find shortest path distances since then
$w_p(e)\geq 0$ for all $e\in E_H$. An example of a feasible price function is $u\mapsto d_H(s,u)$ for any $s\in V_H$.
This assumes that $d_H(s,u) < \infty$ for all $u\in V_H$ which can always be achieved by, say, triangulating
$H$ with edges of sufficiently large length so that shortest paths in
$H$ will not be affected.

A matrix $M=(M_{ij})$ is \emph{totally monotone} if for every $i,i',j,j'$ such
that $i<i'$, $j<j'$ and $M_{ij} \leq M_{ij'}$, we also have $M_{i'j} \leq M_{i'j'}$.
Totally monotone matrices were introduced by Aggarwal et
al. in~\cite{SMAWK}, who gave an algorithm, nicknamed SMAWK, that, given
a totally monotone $n \times m$ matrix $M$,
finds all column minima of $M$ in just $O(n+m)$ time.
A matrix $M=(M_{ij})$ is \emph{convex Monge}  if for every $i,i',j,j'$ such
that $i<i'$, $j<j'$, we have $M_{ij} + M_{i'j'} \geq M_{ij'} +
M_{i'j}$.
It is immediate that if $M$ is convex Monge then it is totally monotone.
Thus SMAWK can be used to find the column minima of a convex
Monge matrix.
\cite{SSSPPlanarNeg} used a generalization of SMAWK to 
so called falling staircase matrices, due to Klawe and Kleitman~\cite{MinimaMatrix},
that finds all column minima in $O(m\alpha(n) + n)$ time, where $\alpha(n)$ is the inverse Ackerman function.

\section{The Algorithm of Klein et al.}\label{sec:KleinAlgo}
In this section, we give an overview of the algorithm of~\cite{SSSPPlanarNeg} before describing our improved
algorithm in Section~\ref{sec:ImprovedAlgo}.

Let $s$ be a vertex of $G$. To find SSSP distances in $G$ with source $s$, the algorithm starts by applying the
cycle separator theorem to $G$. This gives a Jordan curve $C$ which separates $G$ into
two subgraphs, $G_0$ and $G_1$, and there are $O(\sqrt n)$ boundary vertices on $C$.

Let $r$ be any of these boundary vertices. The algorithm consists of five stages:
\paragraph{Recursive call:} SSSP distances in $G_i$ with source $r$ are computed recursively for $i = 0,1$.
\paragraph{Intra-part boundary distances:} The distances in $G_i$ between each pair of boundary vertices of $G_i$ are
computed using the algorithm of~\cite{MultiSrc} for $i = 0,1$. This stage takes $O(n\log n)$ time.
\paragraph{Single-source inter-part boundary distances:} A variant of Bellman-Ford is used to compute SSSP distances
in $G$ from $r$ to all boundary vertices on $C$. The algorithm consists of $O(\sqrt n)$ iterations and each
iteration runs in $O(\sqrt n\alpha(n))$ time using the algorithm of Klawe and Kleitman~\cite{MinimaMatrix}. This stage therefore runs in
$O(n\alpha(n))$ time.
\paragraph{Single-source inter-part distances:} Distances in the previous stage are used to modify $G$ such that
all edge lengths are non-negative without changing the shortest paths. Dijkstra's algorithm is then used in
the modified graph to obtain SSSP distances in $G$ with source $r$. Total running time for this stage is $O(n\log n)$.
\paragraph{Rerooting single-source distances:} A price function is obtained from the computed distances from $r$ in $G$.
This price function is feasible for $G$ and Dijkstra's algorithm is applied to obtain SSSP distances
in $G$ with source $s$ in $O(n\log n)$ time.

\section{Improved Algorithm}\label{sec:ImprovedAlgo}
As can be seen above, the last four stages of the algorithm in~\cite{SSSPPlanarNeg} run in a total of $O(n\log n)$ time.
Since there are $O(\log n)$ recursion levels, the total running time is $O(n\log^2n)$. We now describe how to improve
this time bound.

The idea is to reduce the number of recursion levels by applying the cycle separator theorem of Miller not once but
several times at each node of the recursion tree. More precisely, for a suitable $p$, we obtain an $n/p$-division of $G$ in
$O(n\log n)$ time.
For each region $R_i$ in this $n/p$-division, we pick an arbitrary boundary vertex $r_i$ and recursively compute SSSP
distances in $R_i$ with source $r_i$. This is similar to stage one in the original algorithm except that we recurse on
$O(p)$ regions instead of just two.

We will show how all these recursively computed distances can be used to compute SSSP distances in $G$ with source $s$
in $O(n\log n + np\alpha(n))$ additional time. This bound is no better than the $O(n\log n)$ bound of the original
algorithm but the speed-up comes from the reduced number of recursion levels. Since the size of regions
is reduced by a factor of at least $p$ for each recursion level, the depth of the recursion tree is only $O(\log n/\log p)$.
It follows that the total running time of our algorithm is
\[
  O\left(\frac{\log n}{\log p}(n\log n + np\alpha(n))\right).
\]
To minimize this expression, we set $n\log n = np\alpha(n)$. Solving this, we get $p = \log n/\alpha(n)$
which gives a running time of $O(n\log^2n/\log\log n)$, as requested.

What remains is to show how to compute SSSP distances in $G$ with source $s$ in $O(n\log n + np\alpha(n)) = O(n\log n)$
time, excluding the time for recursive calls.

So assume that we are given an $n/p$-division of $G$ and that for each region $R$, we are given SSSP distances in $R$ with
some boundary vertex of $R$ as source. Note that the number of regions is $O(p)$ and each region contains at most $n/p$
vertices and $O(\sqrt{n/p})$ boundary vertices.

We will assume in the following that no region has holes. Then all its boundary vertices are cyclically
ordered on its external face. We consider holes in Section~\ref{subsec:Holes}.

The remaining part of the algorithm consists of four stages very similar to those in the algorithm of Klein et al. We give an
overview of them here and describe them in greater detail in the subsections below. Each stage takes $O(n\log n)$ time.
\paragraph{Intra-region boundary distances:} For each region $R$, distances in $R$ between each pair of boundary
vertices of $R$ are computed.
\paragraph{Single-source inter-region boundary distances:} Distances in $G$ from an arbitrary boundary vertex $r$ of
an arbitrary region to all boundary vertices of all regions are computed.
\paragraph{Single-source inter-region distances:} Using the distances obtained in the previous stage to obtain
a modified graph, distances in $G$ from $r$ to all vertices of $G$ are computed using Dijkstra's algorithm on the
modified graph.
\paragraph{Rerooting single-source distances:} Identical to the final stage of the original algorithm.

\subsection{Intra-region Boundary Distances}\label{subsec:IntraregBoundaryDists}
Let $R$ be a region. Since $R$ has no holes, we can apply the multiple-source shortest path algorithm of~\cite{MultiSrc} to
$R$ since we have a feasible price function from the recursively computed distances in $R$. Total time for this is
$O(|V_R|\log|V_R|)$ time which is $O(n\log n)$ over all regions.

\subsection{Single-source Inter-region Boundary Distances}\label{subsec:SSIntregBoundaryDists}
Let $r$ be some boundary vertex of some region. We need to find distances in $G$ from $r$ to all boundary vertices of
all regions. To do this, we use a variant of Bellman-Ford similar to that in stage three of the original algorithm.

Let $\mathcal R$ be the set of $O(p)$ regions, let $B\subseteq V$ be the set of boundary vertices over all regions, and let
$b = |B| = O(p\sqrt{n/p}) = O(\sqrt{np})$. Note that a vertex in $B$ may belong to several regions.

Pseudocode of the algorithm is shown in Figure~\ref{fig:Pseudocode}. Notice the similarity with the algorithm
in~\cite{SSSPPlanarNeg} but also an important difference: in~\cite{SSSPPlanarNeg}, each table entry $e_j[v]$ is updated only
once. Here, it may be updated several times in iteration $j$ since more than one region may have $v$ as a boundary vertex. For
$j\geq 1$, the final value of $e_j[v]$ will be
\begin{align}
e_j[v] & = \min_{w\in B_v}\{e_{j-1}[w] + d_R(w,v)\},\label{eqn}
\end{align}
where $B_v$ is the set of boundary vertices of regions having $v$ as boundary vertex.
\begin{figure}
\begin{tabbing}
d\=dd\=\quad\=\quad\=\quad\=\quad\=\quad\=\quad\=\quad\=\quad\=\quad\=\quad\=\quad\=\kill
\>1. \>initialize vector $e_j[v]$ for $j = 0,\ldots, b$ and $v\in B$\\
\>2. \>$e_j[v] := \infty$ for all $v\in B$ and $j = 0,\ldots, b$\\
\>3. \>$e_0[r] := 0$\\
\>4. \> \textbf{for} $j = 1,\ldots,b$\\
\>5. \>\> \textbf{for} each region $R\in\mathcal R$\\
\>6. \>\>\> let $C$ be the cycle defining the boundary of $R$\\
\>7. \>\>\> $e_j[v] := \min\{e_j[v], \min_{w\in V_C}\{e_{j-1}[w] + d_R(w,v)\}\}$ for all $v\in V_C$\\
\>8. \>$D[v] := e_b[v]$ for all $v\in B$
\end{tabbing}
\caption{Pseudocode for single-source inter-region boundary distances algorithm.}\label{fig:Pseudocode}
\end{figure}

To show the correctness of the algorithm, we need the following two lemmas.
\begin{Lem}\label{Lem:BellmanFord1}
Let $P$ be a simple $r$-to-$v$ shortest path in $G$ where $v\in B$. Then $P$ can be decomposed into at
most $b$ subpaths $P = P_1P_2P_3\ldots$, where the endpoints of each subpath $P_i$ are boundary vertices
and $P_i$ is a shortest path in some region of $\mathcal R$.
\end{Lem}
\begin{proof}
$P$ is simple so it can use a boundary vertex at most once. There are $b$ boundary vertices in total. A path
can only enter and leave a region through boundary vertices of that region.
\end{proof}
\begin{Lem}\label{Lem:BellmanFord2}
After iteration $j$ of the algorithm in Figure~\ref{fig:Pseudocode}, $e_j[v]$ is the length of a
shortest path in $G$ from $r$ to $v$ that can be decomposed into at most $j$ subpaths $P = P_1P_2P_3\ldots P_j$,
where the endpoints of each subpath $P_i$ are boundary vertices and $P_i$ is a shortest path in a region of $\mathcal R$.
\end{Lem}
\begin{proof}
The proof is by induction on $j$. We omit it since it is similar to that in~\cite{SSSPPlanarNeg}.
\end{proof}
It follows from~(\ref{eqn}) and Lemmas~\ref{Lem:BellmanFord1} and~\ref{Lem:BellmanFord2} that after $b$ iterations, $D[v]$ holds
the distance in $G$ from $r$ to $v$ for all $v\in B$. This shows the correctness of our algorithm.

Line $7$ can be executed in $O(|V_C|\alpha(|V_C|))$ time using ideas
from~\cite{SSSPPlanarNeg} (since all boundary vertices of $R$ are cyclically ordered on its external face) and the fact that
$d_R(w,v)$ has been precomputed in the previous stage for all $v,w\in V_C$.
Thus, each iteration of lines $4$--$7$ takes $O(b\alpha(n))$ time, giving a total running
time for this stage of $O(b^2\alpha(n)) = O(np\alpha(n))$. Recalling that $p = \log n/\alpha(n)$, this
bound is $O(n\log n)$, as requested.

\subsection{Single-source Inter-region Distances}
To compute distances in $G$ from boundary vertex $r$ to all vertices of $G$ we consider one region at a time. So let
$R$ be a region. We need to compute distances in $G$ from $r$ to each vertex of $R$.

Let $R'$ be the graph obtained from $R$ by adding a new vertex $r'$ and an edge from $r'$ to each boundary vertex of $R$; the length
of this edge is equal to the distance in $G$ from $r$ to the boundary vertex. Note that $d_G(r,v) = d_{R'}(r',v)$ for all
$v\in V_R$. Also note that $R'$ has $O(|V_R|)$ vertices and edges and can be computed in
$O(|V_R|)$ time, given the distances computed in the previous stage. We need to find distances in $R'$ from
$r'$ to each vertex of $V_R$.

Let $r_R$ be the boundary vertex of $R$ for which distances in $R$ from $r_R$ to all vertices of $R$ have been
recursively computed. We define a price function $\phi$ for $R'$ as follows. Let $B_R$ be the set of boundary vertices
of $R$ and let $D = \max\{d_R(r_R,b) - d_G(r,b) | b\in B_R\}$. Then for all $v\in V_{R'}$,
\[
  \phi(v) = \left\{\begin{array}{ll} d_R(r_R,v) & \mbox{ if } v\neq r' \\
                                     D          & \mbox{ if } v = r'.
                   \end{array}\right.
\]
\begin{Lem}\label{Lem:Price}
Function $\phi$ defined above is a feasible price function for $R'$.
\end{Lem}
\begin{proof}
Let $e = (u,v)$ be an edge of $R'$. By construction, no edges end in $r'$ so $v\neq r'$. If $u\neq r'$ then
$\phi(u) + l(e) - \phi(v) = d_R(r_R,u) + l(u,v) - d_R(r_R,v) \geq 0$ by the triangle inequality so
assume that $u = r'$. Then $v\in B_R$ so
$\phi(u) + l(e) - \phi(v) = D + d_G(r,v) - d_R(r_R,v)\geq 0$ by definition of $D$. This shows the lemma.
\end{proof}
Price function $\phi$ can be computed in time linear in the size of $R$ and
Lemma~\ref{Lem:Price} implies that Dijkstra's algorithm can be applied to compute distances in $R'$
from $r'$ to all vertices of
$V_R$ in $O(|V_R|\log|V_R|)$ time. Over all regions, this is $O(n\log n)$, as requested.

We omit the description of the last stage where single-source distances are rerooted to source $s$ since it is identical
to the last stage of the original algorithm. We have shown that all stages run in $O(n\log n)$ time and it follows that
the total running time of our algorithm is $O(n\log^2n/\log\log n)$. It remains to deal with holes in regions.

\subsection{Dealing with Holes}\label{subsec:Holes}
In Sections~\ref{subsec:IntraregBoundaryDists} and~\ref{subsec:SSIntregBoundaryDists}, we needed the assumption
that no region has holes. In this section, we remove this restriction.
As mentioned in Section~\ref{sec:DefsRes}, we may assume w.l.o.g.\ that each region of $\mathcal R$ has at most a
constant $h$ number of holes.

\paragraph{Intra-region boundary distances:}
Let us first show how to compute intra-region boundary distances when regions have holes. The reason why it works
in $O(n\log n)$ time in Section~\ref{subsec:IntraregBoundaryDists} is that all boundary vertices of each region are on the
external face, allowing us to apply the multiple-source shortest path algorithm of~\cite{MultiSrc}.

Now, consider a region $R$. If we apply~\cite{MultiSrc} to $R$ we get distances from boundary vertices on the external face
of $R$ to all boundary vertices of $R$. This is not enough. We also need distances from boundary vertices belonging
to the holes of $R$.

So consider one of the holes of $R$. We can transform $R$ in linear time such that this hole becomes the external face of $R$.
Having done this transformation, we can apply the algorithm of~\cite{MultiSrc} to get distances from boundary vertices of
this hole to all boundary vertices of $R$. If we repeat this for all holes, we get distances in $R$ between all pairs of
boundary vertices of $R$ in time $O(|V_R|\log |V_R| + h|V_R|\log |V_R|) = O(|V_R|\log |V_R|)$ time. Thus, the time bound in
Section~\ref{subsec:IntraregBoundaryDists} still holds when regions have holes.

\paragraph{Single-source inter-region boundary distances:}
What remains is the problem of computing single-source inter-region boundary distances when regions have holes. Let $C$
be the external face of region $R$. Let $H_R$ be the directed graph having the boundary vertices of $R$ as vertices and
having an edge $(u,v)$ of length $d_R(u,v)$ between each pair of vertices $u$ and $v$.

Line 7 in Figure~\ref{fig:Pseudocode} relaxes all edges in $H_R$ having both endpoints on $C$.
We need to relax all edges of $H_R$. In the following, when we say that we relax edges of $R$, we really refer to the
edges of $H_R$.

To relax the edges of $R$, we consider each pair of cycles $(C_1,C_2)$, where $C_1$ and $C_2$ are $C$ or a hole, and
we relax all edges starting in $C_1$ and ending in $C_2$. This will cover all edges we need to relax.

Since the number of choices of $(C_1,C_2)$ is $O(h^2) = O(1)$, it suffices to show that in a single iteration, the time
to relax all edges starting in $C_1$ and ending in $C_2$ is
$O((|V_{C_1}| + |V_{C_2}|)\alpha(|V_{C_1}| + |V_{C_2}|))$, with $O(|V_R|\log |V_R|)$ preprocessing time.

We may assume that $C_1 \neq C_2$, for otherwise we can relax edges as described in
Section~\ref{subsec:SSIntregBoundaryDists}.

In the following, we define graphs, obtained from $R$, that are needed in our algorithm. It is assumed that these graphs
are constructed in a preprocessing step. Later, we bound the time to construct them.

We transform $R$ in such a way that $C_1$ is the external face of $R$ and $C_2$ is a hole of $R$.
We may assume that there is a shortest path in $R$ between every ordered pair of vertices, say, by adding a pair of oppositely
directed edges between each consecutive pair of vertices of $C_i$ in some simple walk of $C_i$, $i = 1,2$ (if an edge already
exists, a new edge is not added). The lengths of the new edges are chosen sufficiently large so that shortest paths in $R$ and
their lengths will not change. Where appropriate, we will regard $R$ as some fixed planar embedding of that region.

Before proceeding, let us give some additional definitions. 
We say that an edge $e = (u,v)$ with exactly one endpoint on path $P$
\emph{emanates right (left) of $P$} if (a) $e$ is directed away from $P$, and
(b) $e$ is to the right (left) of $P$ in the direction of $P$.
If $e$ is directed towards $P$, then we say that $e$ \emph{enters $P$ from the right (left)}
if $(v,u)$ emanates right (left) of $P$.
We extend these definitions to paths and say, e.g., that a  path $Q$
emanates right of path $P$ if there is an edge of $Q$
that emanates right of $P$.

Now, let $r_1\in V_{C_1}$ and let $T$ be a shortest path tree in $R$ from $r_1$ to all vertices of $C_2$. Let $P$ be a simple path in
$T$ from
$r_1$ to some leaf $r_2\in V_{C_2}$. Define $\overleftarrow{E}$ resp.\ $\overrightarrow{E}$ as the set of edges that
either emanate left resp.\ right of $P$ or enter $P$ from the left resp.\ right.

Now, take a copy $R_P$ of $R$ and remove $P$ and all edges incident to $P$ in $R_P$. Add two copies, $\overleftarrow{P}$ and
$\overrightarrow{P}$, of $P$ to $R_P$. Connect path $\overleftarrow{P}$ resp.\ $\overrightarrow{P}$ to the rest of $R_P$ by attaching
the edges of $\overleftarrow{E}$ resp.\ $\overrightarrow{E}$ to the path, see Figure~\ref{fig:OneCut}. If $(u,v)\in E_R$, where
$(v,u)\in E_P$, we add $(u,v)$ to $\overleftarrow P$ and $\overrightarrow P$ in $R_P$. We extend this construction to paths from
$C_1$ to $C_2$ other than $P$.

In order to relax edges from boundary vertices of $C_1$ to boundary vertices of $C_2$ in $R$, the first step is to relax edges
in $R_P$, defined in the following.
\begin{figure}
\centerline{\scalebox{0.6}{\input{OneCut.pstex_t}}}
\caption{Region $R_P$ is obtained from $R$ essentially by cutting open at $P$ the ``ring'' bounded by $C_1$ and $C_2$.}
\label{fig:OneCut}
\end{figure}
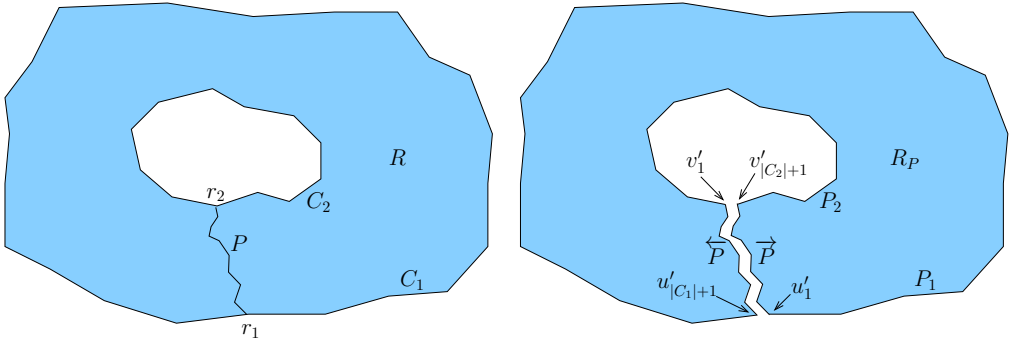

Note that a simple, say counter-clockwise, walk $(u_1 = r_1),u_2,\ldots,u_{|C_1|},(u_{|C_1|+1} = r_1)$ of $C_1$ in $R$ starting
and ending in $r_1$ corresponds to a simple path $P_1 = u_1',\ldots,u_{|C_1|+1}'$ in $R_P$. In the following, we identify
$u_i$ with $u_i'$ for $i = 2,\ldots,|C_1|$. Vertex $u_1 = r_1$ in $R$ corresponds to two vertices in $R_P$, namely $u_1'$ and
$u_{|C_1|+1}$. We will identify both of these vertices with $r_1$.

A simple, say clockwise, walk of $C_2$ in $R$ from $r_2$ to $r_2$ corresponds to a simple path
$P_2 = v_1',\ldots,v_{|C_2|+1}'$ in $R_P$. We make a similar identification between vertices of $C_2$ and $P_2$.

In the following, when we e.g.\ say that we relax all edges in $R_P$ starting in vertices of $C_1$ and ending in
vertices of $C_2$, we really refer to edges starting in the corresponding vertices of $P_1$ and ending in the corresponding
vertices of $P_2$. More precisely, suppose we are in iteration $j$. Then relaxing an edge from a $u\in V_{C_1}\setminus\{r_1\}$
to a $v\in V_{C_2}\setminus\{r_2\}$ in $R_P$ means updating
\[
  e_j[v] := \min\{e_j[v], e_{j-1}[u] + d_{R_P}(u',v')\}.
\]
If $u = r_1$, we relax w.r.t.\ both $u_1'$ and $u_{|C_1|+1}'$ and if $v = r_2$, we relax w.r.t.\ both $v_1'$ and $v_{|C_2|+1}'$.
We extend these definitions to graphs with a structure similar to $R_P$.

As the following lemma shows, relaxing edges in $R_P$ can be done efficiently by exploiting the cyclic order of boundary vertices
of $R_P$ as we did above for regions with no holes.
\begin{Lem}\label{Lem:Relax}
Relaxing all edges from $V_{C_1}$ to $V_{C_2}$ in $R_P$ can be done in $O(|V_{C_1}| + |V_{C_2}|)$
time in any iteration of Bellman-Ford.
\end{Lem}
\begin{proof}
Let paths $P_1$ and $P_2$ in $R_P$ be defined as above.
Consider iteration $j$. Define a $|P_1|\times |P_2|$ matrix $A$ with elements $A_{kl} = e_{j-1}[u_k] + d_{R_P}(u_k',v_l')$.
Observe that relaxing all edges from $V_{C_1}$ to $V_{C_2}$ in $R_P$ is equivalent to finding all column-minima of $A$
(compare this to~\cite{SSSPPlanarNeg}).

Now, since $P_1\overleftarrow{P}P_2\overrightarrow{P}$ is a cycle, it follows easily from results of~\cite{SSSPPlanarNeg}
that for $1\leq k\leq k'\leq |P_1|$ and $1\leq l\leq l'\leq |P_2|$, $A_{kl} + A_{k'l'} \geq A_{kl'} + A_{k'l}$.
Hence $A$ is Monge (see~\cite{SSSPPlanarNeg}), so by~\cite{SMAWK},
its column-minima can be found in $O(|V_{C_1}| + |V_{C_2}|)$ time.
\end{proof}

Unfortunately, relaxing edges between boundary vertices in $R_P$ will
not suffice since shortest paths in $R$ that cross $P$ are not
represented in $R_P$.
To overcome this obstacle we identify two particular paths $P_r$ and
$P_\ell$ such that for any $u\in C_1, v \in C_2$ there exists a shortest
path in $R$ that does not cross both $P_{r}$ and $P_\ell$. Then, 
relaxing all edges between boundary vertices  once in $R_{P_r}$ and
once in $R_{P_\ell}$ suffices to compute shortest path distances in $R$. We now
formalize this idea.

Let us define what we mean when we say that path $Q = q_1,
q_2, q_3, \dots$ crosses path $P$. Let $out_0$ be the smallest index
such that $q_{out_0}$ does not belong to $P$. We recursively define
$in_i$ to be smallest index greater than  $out_{i-1}$ such that 
$q_{in_i}$ belongs to $P$, and $out_i$ to be smallest index greater than $in_i$ such that 
$q_{in_i}$ does not belong to $P$. We say that $Q$ crosses $P$ from
the right (left) with entry vertex $v_{in}$ and exit vertex $v_{out}$
if (a) $v_{in} = q_{in_i}$ and $v_{out} = q_{out_i-1}$ for some $i>0$
and (b) $q_{in_i-1}q_{in_i}$ enters $P$ from the right (left) and (c)
$q_{out_i-1}q_{out}$ emanates left (right) of $P$.

A \emph{rightmost (leftmost)} path $P$ in $T$ is a  path such that no
other path $Q$ in $T$ emanates right (left) of $P$.
Let $P_r$ and $P_\ell$ be the rightmost and leftmost root-to-leaf
simple paths in $T$, respectively; see Figure~\ref{fig:TwoCuts}(a).
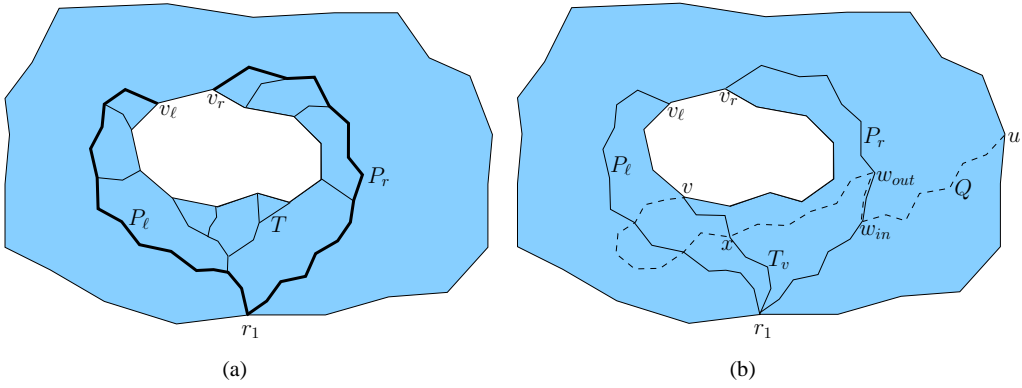
\begin{figure}
\centerline{\scalebox{0.6}{\input{TwoCuts.pstex_t}}}
\caption{(a): The rightmost root-to-leaf simple path $P_r$ and the leftmost
         root-to-leaf simple path $P_\ell$ in $T$. (b): In the proof
         of Lemma~\ref{Lem:NoTwoCross}, if $Q$ first crosses $P_r$ from
         right to left and then crosses $P_\ell$ from right to left then
         there is a $u$-to-$v$ shortest path in $R$ that does not cross
         $P_\ell$.}
\label{fig:TwoCuts}
\end{figure}
Let $v_r\in C_2$ and $v_\ell \in C_2$
denote the leaves of $P_r$ and $P_\ell$, respectively.

\begin{Lem}\label{Lem:NoTwoCross}
For any $u\in V_{C_1}$ and any $v\in V_{C_2}$, there is a simple shortest path in $R$ from $u$ to $v$ which
does not cross both $P_r$ and $P_\ell$.
\end{Lem}
\begin{proof}
Suppose for the sake of contradiction that there are vertices
$u\in V_{C_1}$ and $v\in V_{C_2}$ such that any simple shortest path
in $R$ from $u$ to $v$ crosses both $P_r$ and $P_\ell$.
Let $Q$ be a simple $u$-to-$v$ shortest path in $R$ such that the sum
of the number of times it crosses $P_r$ and the number of times it crosses
$P_\ell$ is minimal. 

Assume that $Q$ crosses $P_r$ first. The case where $Q$ crosses $P_\ell$
first is symmetric.
Let $w_{in}$ and $w_{out}$ be the entry and exit vertices
of the first crossing, see Figure~\ref{fig:TwoCuts}(b).
If $Q$ first crosses $P_r$ from left to right, then observe that it must cross $P_\ell$
at the same vertices. In fact, it must be that all root-to-leaf paths in $T$ coincide
until $w_{out}$ and that $Q$ crosses all of them.
In particular, $Q$ crosses the root-to-$v$  path in $T$, which we
denote by $T_v$. 
Since $T_v$ does not cross $P_r$, the path $Q[u,w_{out}]T_v[w_{out},v]$ is a shortest
path from $u$ to $v$ in $R$ that does not cross $P_r$. But this contradicts our
assumption above.

We conclude that $Q$ first crosses $P_r$ from right to left. Consider
the path $S = Q[u,w_{out}]P_r[w_{out},v_r]$. 
We claim that $Q$ does not cross $S$. To see this,
assume the contrary and let
$w'$ denote the exit point corresponding to the crossing. Since $Q$ is simple, $w' \notin Q[u,w_{out}]$.
So $w' \in P_r[w_{out},v_r]$, but
then $Q[u,w_{out}]P_r[w_{out},w']Q[w',v]$ is a shortest path from
$u$ to $v$ in $R$ that crosses $P_r$ and $P_\ell$ fewer times than $Q$. But this
contradicts the minimality of $Q$.

Since $Q$ first crosses $P_r$ from right to left and never crosses $S$,
its first crossing with $P_\ell$ must be right-to-left as well, see
Figure~\ref{fig:TwoCuts}(b). This
implies that $Q$ enters all root-to-leaf paths in $T$ 
before (not strictly before) it enters $P_\ell$. In particular, $Q$
enters $T_v$. Let $x$ be the entry vertex. Then $Q[u,x]T_v[x,v]$
is a $u$-to-$v$ shortest path in $R$ that does not
cross $P_\ell$, contradicting our assumption.
\end{proof}

\paragraph{The algorithm:}
We can now describe our Bellman-Ford algorithm to relax all edges from vertices of $C_1$ to vertices of $C_2$.
Pseudocode is shown in Figure~\ref{fig:Pseudocode2}.

Assume that $R_{P_l}$ and $R_{P_r}$ and distances between pairs of boundary vertices in these graphs have been precomputed.

In each iteration $j$, we relax edges from vertices of $V_{C_1}$ to all $v\in V_{C_2}$ in
$R_{P_\ell}$ and in $R_{P_r}$ (lines $9$ and $10$). Lemma~\ref{Lem:NoTwoCross} implies that this corresponds to relaxing
all edges in $R$ from vertices of $V_{C_1}$ to vertices of $V_{C_2}$.  By the results
in Section~\ref{subsec:SSIntregBoundaryDists}, this suffices to show the correctness of the algorithm.

Lemma~\ref{Lem:Relax} shows that lines $9,10$ can each be implemented to run in $O(|V_{C_1}| + |V_{C_2}|)$
time. Thus, each iteration of lines $6$--$10$ takes
$O((|V_{C_1}| + |V_{C_2}|)\alpha(|V_{C_1}| + |V_{C_2}|))$ time, as requested.

It remains to show that $R_{P_r}$ and $R_{P_\ell}$ and distances between boundary vertices in these graphs can be precomputed in
$O(|V_R|\log|V_R|)$ time 

\begin{figure}
\begin{tabbing}
d\=ddd\=\quad\=\quad\=\quad\=\quad\=\quad\=\quad\=\quad\=\quad\=\quad\=\quad\=\quad\=\kill
\>1.  \>initialize vector $e_j[v]$ for $j = 0,\ldots, b$ and $v\in B$\\
\>2.  \>$e_j[v] := \infty$ for all $v\in B$ and $j = 0,\ldots, b$\\
\>3.  \>$e_0[r] := 0$\\
\>4.  \> \textbf{for} $j = 1,\ldots,b$\\
\>5.  \>\> \textbf{for} each region $R\in\mathcal R$\\
\>6.  \>\>\> \textbf{for} each pair of cycles, $C_1$ and $C_2$, defining the boundary of $R$\\
\>7.  \>\>\>\> \textbf{if} $C_1 = C_2$, relax edges from $C_1$ to $C_2$ as in~\cite{SSSPPlanarNeg}\\
\>8.  \>\>\>\> \textbf{else} (assume $C_1$ is external and that $d_{R_{P_r}}$ and $d_{R_{P_\ell}}$ have been precomputed)\\
\>9.  \>\>\>\>\> $e_j[v] := \min\{e_j[v],\min_{w\in V_{C_1}}\{e_{j-1}[w] + d_{R_{P_r}}(w,v)\}\}$
                 for all $v\in V_{C_2}$\\
\>10. \>\>\>\>\> $e_j[v] := \min\{e_j[v],\min_{w\in V_{C_1}}\{e_{j-1}[w] + d_{R_{P_\ell}}(w,v)\}\}$
                 for all $v\in V_{C_2}$\\
\>11. \> $D[v] := e_b[v]$ for all $v\in B$
\end{tabbing}
\caption{Pseudocode for the Bellman-Ford variant that handles regions with holes.}\label{fig:Pseudocode2}
\end{figure}

Shortest path tree $T$ in $R$ with source $r_1$ can be found in $O(|V_R|\log|V_R|)$ time with Dijkstra using the
feasible price function $\phi$ obtained from the recursively computed
distances in $R$. Given $T$, we can find its rightmost
path in $O(|V_R|)$ time by starting at the root $r_1$. When entering a
vertex $v$ using the edge $uv$, leave that vertex on the edge that
comes after $vu$ in counterclockwise order. Computing
$R_{P_r}$ given $P_r$ also takes
$O(|V_R|)$ time. We can next apply Klein's algorithm~\cite{MultiSrc} to compute distances between all pairs of boundary vertices
in $R_{P_r}$ in $O(|V_R|\log|V_R|)$ time (here, we use the non-negative edge lengths in $R$ defined by the reduced cost function
induced by $\phi$). We similarly compute $P_\ell$ and pairwise distances
between boundary vertices in $R_{P_\ell}$.

We can now state our result.
\begin{theorem}\label{Thm:MainRes}
Given a planar directed graph $G$ with real edge lengths and no negative cycles and given a source vertex $s$, we can find
SSSP distances in $G$ with source $s$ in $O(n\log^2n/\log\log n)$ time and linear space.
\end{theorem}
\begin{proof}
We gave the bound on running time above. To bound the space, first note that finding an $n/p$-division of $G$ using the
algorithm of~\cite{Fakcharoenphol} requires $O(n)$ space. Klein's algorithm~\cite{MultiSrc} and Dijkstra also has linear
space requirement. The recursively computed distances take up a total of $O(p\frac n p) = O(n)$ space. In the intra-region
boundary distances stage, the total memory spent on storing distances is $O(p(\sqrt{n/p})^2) = O(n)$.

In the single-source inter-region boundary distances stage, we need to bound the space for our Bellman-Ford variant.
The size of each table is $O(b) = O(n)$. Since we only need to keep tables from the current and previous
iteration in memory, Bellman-Ford uses $O(n)$ space.
It is easy to see that the last two stages use $O(n)$ space. Hence the entire algorithm has linear space requirement.
\end{proof}

\section{Concluding Remarks}\label{sec:ConclRem}
We gave a linear space algorithm for single-source shortest path distances in a planar directed graph with arbitrary real
edge lengths and no negative cycles. Running time is $O(n\log^2n/\log\log n)$, an improvement of a previous bound
by a factor of $\log\log n$. As corollaries, bipartite planar perfect matching, feasible flow, and feasible
circulation in planar graphs can be solved in $O(n\log^2n/\log\log n)$ time.

It would be interesting to consider other types of graphs. Results from~\cite{BoundedGenus} seem to imply that our
algorithm generalizes to bounded genus graphs.

Finding the true complexity of the problem remains open since there is still a gap between our upper bound and the
linear lower bound. Is $O(n\log n)$ time achievable?


\newpage
\section*{Appendix}
The proofs left out or only sketched in the main paper were very similar to those in~\cite{SSSPPlanarNeg}. For completeness, we
give them here.

\subsection*{Proof of Lemma~\ref{Lem:BellmanFord2}}
We need to show that after iteration $j$ of the algorithm in Figure~\ref{fig:Pseudocode}, $e_j[v]$ is the length of a
shortest path in $G$ from $r$ to $v$ that can be decomposed into at most $j$ subpaths $P = P_1P_2P_3\ldots P_j$,
where the endpoints of each subpath $P_i$ are boundary vertices and $P_i$ is a shortest path in a region of $\mathcal R$.

The proof is by induction on $j\geq 0$ (and very similar to the proof of Lemma 4.2 in~\cite{SSSPPlanarNeg}). When
$j = 0$, $e_j[r] = 0$ and $e_j[v] = \infty$ for all $v\in B\setminus\{r\}$ after line $3$ and the base case holds.

Suppose $j > 0$ and that the lemma holds for $j - 1$. Consider a shortest path $P$ in $G$ from $r$ to a $v\in B$ that can be
decomposed into subpaths $P_1P_2P_3\ldots P_j$, where the endpoints of each subpath $P_i$ are boundary vertices and $P_i$
is a shortest path in a region of $\mathcal R$. We need to show that after iteration $j$, $e_j[v]$ is the length of $P$.

Subpath $P' = P_1P_2\ldots,P_{j-1}$ is a shortest path in $G$ from $r$ to a $w\in B$ which can be decomposed into at most
$j-1$ subpaths as above. Furthermore, there is a region $R\in\mathcal R$
such that $v$ and $w$ are boundary vertices of $R$ and $P_j$ is a shortest path in $R$ from $w$ to $v$.

At some point in iteration $j$, we reach line $7$ with $C$ being the cycle defining the boundary
of $R$ and $v,w\in V_C$. By the induction hypothesis, $e_{j-1}[w]$ is the length of $P'$. Since
$e_j[v]$ is set to a value of at most $e_{j-1}[w] + d_R(w,v)$, $e_j[v]$ is at most the length of $P$.

Let us show the other inequality. For any $w\in B$, $e_{j-1}[w]$ is clearly the length of some path in $G$ from $r$ to $w$
that can be decomposed into at most $j-1$ subpaths, where each subpath is a shortest path in a region between two boundary
vertices of that region. Hence, when $e_j[v]$ is updated in line $7$, its value is the length of some path in $G$ from $r$ to
$v$ that can be decomposed into at most $j$ such subpaths. This shows that $e_j[v]$ is at least the length of $P$, completing
the proof.

\subsection*{Proof of Lemma~\ref{Lem:Relax}}
We need to show that relaxing all edges from $V_{C_1}$ to $V_{C_2}$ in $R_P$ can be done in
$O(|V_{C_1}| + |V_{C_2}|)$ time in any iteration of Bellman-Ford.

We only sketched a proof in the main paper. Let paths $P_1$ and $P_2$ and $|P_1|\times|P_2|$ matrix $A$ be defined as in the
proof sketch. We need to show that the column-minima of $A$ can be found in
$O(|V_{C_1}| + |V_{C_2}|)$ time. As shown in the main paper, this amounts to showing that 
for $1\leq k < k'\leq |P_1|$ and $1\leq l < l'\leq f(i)$, we have $A_{kl} + A_{k'l'} \geq A_{kl'} + A_{k'l}$.

Since $P_1\overleftarrow{P}P_2\overrightarrow{P}$ is
a cycle and since $R_P$ is planar, any pair of paths in $R_P$ from $u_{k}'$ to $v_{l}'$ and
from $u_{k'}'$ to $v_{l'}'$ must intersect in some $w\in V_{R_P}$, see Figure~\ref{fig:Monge}.
\begin{figure}
\centerline{\scalebox{0.6}{\input{Monge.pstex_t}}}
\caption{The situation in the proof of Lemma~\ref{Lem:Relax}. Any pair of paths in $R_P$ from $u_{k}'$ to $v_{l}'$ and
from $u_{k'}'$ to $v_{l'}'$ must intersect in some $w\in V_{R_P}$.}
\label{fig:Monge}
\end{figure}
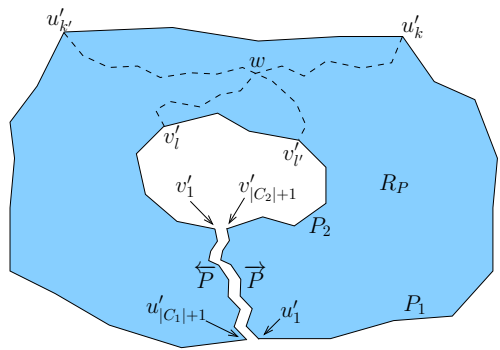
Let $b_k = e_{j-1}[u_{k}]$ and $b_k' = e_{j-1}[u_{k'}]$ (recall that we identified $u_k$ with $u_k'$ and
$u_{k'}$ with $u_{k'}'$). Then
\begin{align*}
A_{kl} + A_{k'l'} &    = (b_k + d_{R_P}(u_k',w) + d_{R_P}(w,v_l')) + (b_{k'} + d_{R_P}(u_{k'}',w) + d_{R_P}(w,v_{l'}'))\\
                  &    = (b_k + d_{R_P}(u_k',w) + d_{R_P}(w,v_{l'}')) + (b_{k'} + d_{R_P}(u_{k'}',w) + d_{R_P}(w,v_l'))\\
                  & \geq (b_k + d_{R_P}(u_k',v_{l'}')) + (b_{k'} + d_{R_P}(u_{k'}',v_l'))\\
                  &    = A_{kl'} + A_{k'l},
\end{align*}
as requested.
\end{document}

%% file: OneCut.pstex_t
\begin{picture}(0,0)%
\includegraphics{OneCut.pstex}%
\end{picture}%
\setlength{\unitlength}{4144sp}%
\begingroup\makeatletter\ifx\SetFigFont\undefined%
\gdef\SetFigFont#1#2#3#4#5{%
  \reset@font\fontsize{#1}{#2pt}%
  \fontfamily{#3}\fontseries{#4}\fontshape{#5}%
  \selectfont}%
\fi\endgroup%
\begin{picture}(10024,3388)(5074,-3572)
\put(7328,-2665){\makebox(0,0)[lb]{\smash{{\SetFigFont{14}{16.8}{\familydefault}{\mddefault}{\updefault}{\color[rgb]{0,0,0}$P$}%
}}}}
\put(7441,-3499){\makebox(0,0)[lb]{\smash{{\SetFigFont{14}{16.8}{\familydefault}{\mddefault}{\updefault}{\color[rgb]{0,0,0}$r_1$}%
}}}}
\put(12559,-2798){\makebox(0,0)[lb]{\smash{{\SetFigFont{14}{16.8}{\familydefault}{\mddefault}{\updefault}{\color[rgb]{0,0,0}$\overrightarrow{P}$}%
}}}}
\put(12059,-2797){\makebox(0,0)[lb]{\smash{{\SetFigFont{14}{16.8}{\familydefault}{\mddefault}{\updefault}{\color[rgb]{0,0,0}$\overleftarrow{P}$}%
}}}}
\put(8911,-1816){\makebox(0,0)[lb]{\smash{{\SetFigFont{14}{16.8}{\familydefault}{\mddefault}{\updefault}{\color[rgb]{0,0,0}$R$}%
}}}}
\put(13906,-1816){\makebox(0,0)[lb]{\smash{{\SetFigFont{14}{16.8}{\familydefault}{\mddefault}{\updefault}{\color[rgb]{0,0,0}$R_P$}%
}}}}
\put(9024,-3023){\makebox(0,0)[lb]{\smash{{\SetFigFont{14}{16.8}{\familydefault}{\mddefault}{\updefault}{\color[rgb]{0,0,0}$C_1$}%
}}}}
\put(14150,-3023){\makebox(0,0)[lb]{\smash{{\SetFigFont{14}{16.8}{\familydefault}{\mddefault}{\updefault}{\color[rgb]{0,0,0}$P_1$}%
}}}}
\put(13204,-2243){\makebox(0,0)[lb]{\smash{{\SetFigFont{14}{16.8}{\familydefault}{\mddefault}{\updefault}{\color[rgb]{0,0,0}$P_2$}%
}}}}
\put(8091,-2243){\makebox(0,0)[lb]{\smash{{\SetFigFont{14}{16.8}{\familydefault}{\mddefault}{\updefault}{\color[rgb]{0,0,0}$C_2$}%
}}}}
\put(12924,-3090){\makebox(0,0)[lb]{\smash{{\SetFigFont{14}{16.8}{\familydefault}{\mddefault}{\updefault}{\color[rgb]{0,0,0}$u_1'$}%
}}}}
\put(12504,-1835){\makebox(0,0)[lb]{\smash{{\SetFigFont{14}{16.8}{\familydefault}{\mddefault}{\updefault}{\color[rgb]{0,0,0}$v_{|C_2| + 1}'$}%
}}}}
\put(11878,-1836){\makebox(0,0)[lb]{\smash{{\SetFigFont{14}{16.8}{\familydefault}{\mddefault}{\updefault}{\color[rgb]{0,0,0}$v_1'$}%
}}}}
\put(11584,-3042){\makebox(0,0)[lb]{\smash{{\SetFigFont{14}{16.8}{\familydefault}{\mddefault}{\updefault}{\color[rgb]{0,0,0}$u_{|C_1| + 1}'$}%
}}}}
\put(7094,-2138){\makebox(0,0)[lb]{\smash{{\SetFigFont{14}{16.8}{\familydefault}{\mddefault}{\updefault}{\color[rgb]{0,0,0}$r_2$}%
}}}}
\end{picture}%

%% file: TwoCuts.pstex_t
\begin{picture}(0,0)%
\includegraphics{TwoCuts.pstex}%
\end{picture}%
\setlength{\unitlength}{4144sp}%
\begingroup\makeatletter\ifx\SetFigFont\undefined%
\gdef\SetFigFont#1#2#3#4#5{%
  \reset@font\fontsize{#1}{#2pt}%
  \fontfamily{#3}\fontseries{#4}\fontshape{#5}%
  \selectfont}%
\fi\endgroup%
\begin{picture}(10024,3812)(-36,-3996)
\put(2331,-3505){\makebox(0,0)[lb]{\smash{{\SetFigFont{14}{16.8}{\familydefault}{\mddefault}{\updefault}{\color[rgb]{0,0,0}$r_1$}%
}}}}
\put(7441,-3499){\makebox(0,0)[lb]{\smash{{\SetFigFont{14}{16.8}{\familydefault}{\mddefault}{\updefault}{\color[rgb]{0,0,0}$r_1$}%
}}}}
\put(9973,-1578){\makebox(0,0)[lb]{\smash{{\SetFigFont{14}{16.8}{\familydefault}{\mddefault}{\updefault}{\color[rgb]{0,0,0}$u$}%
}}}}
\put(2625,-2484){\makebox(0,0)[lb]{\smash{{\SetFigFont{14}{16.8}{\familydefault}{\mddefault}{\updefault}{\color[rgb]{0,0,0}$T$}%
}}}}
\put(3589,-1987){\makebox(0,0)[lb]{\smash{{\SetFigFont{14}{16.8}{\familydefault}{\mddefault}{\updefault}{\color[rgb]{0,0,0}$P_r$}%
}}}}
\put(1203,-2414){\makebox(0,0)[lb]{\smash{{\SetFigFont{14}{16.8}{\familydefault}{\mddefault}{\updefault}{\color[rgb]{0,0,0}$P_\ell$}%
}}}}
\put(8494,-2512){\makebox(0,0)[lb]{\smash{{\SetFigFont{14}{16.8}{\familydefault}{\mddefault}{\updefault}{\color[rgb]{0,0,0}$w_{in}$}%
}}}}
\put(8661,-1972){\makebox(0,0)[lb]{\smash{{\SetFigFont{14}{16.8}{\familydefault}{\mddefault}{\updefault}{\color[rgb]{0,0,0}$w_{out}$}%
}}}}
\put(7115,-2659){\makebox(0,0)[lb]{\smash{{\SetFigFont{14}{16.8}{\familydefault}{\mddefault}{\updefault}{\color[rgb]{0,0,0}$x$}%
}}}}
\put(6722,-2092){\makebox(0,0)[lb]{\smash{{\SetFigFont{14}{16.8}{\familydefault}{\mddefault}{\updefault}{\color[rgb]{0,0,0}$v$}%
}}}}
\put(7582,-2819){\makebox(0,0)[lb]{\smash{{\SetFigFont{14}{16.8}{\familydefault}{\mddefault}{\updefault}{\color[rgb]{0,0,0}$T_v$}%
}}}}
\put(6575,-1332){\makebox(0,0)[lb]{\smash{{\SetFigFont{14}{16.8}{\familydefault}{\mddefault}{\updefault}{\color[rgb]{0,0,0}$v_\ell$}%
}}}}
\put(7095,-1179){\makebox(0,0)[lb]{\smash{{\SetFigFont{14}{16.8}{\familydefault}{\mddefault}{\updefault}{\color[rgb]{0,0,0}$v_r$}%
}}}}
\put(8529,-1594){\makebox(0,0)[lb]{\smash{{\SetFigFont{14}{16.8}{\familydefault}{\mddefault}{\updefault}{\color[rgb]{0,0,0}$P_r$}%
}}}}
\put(6009,-1881){\makebox(0,0)[lb]{\smash{{\SetFigFont{14}{16.8}{\familydefault}{\mddefault}{\updefault}{\color[rgb]{0,0,0}$P_\ell$}%
}}}}
\put(9443,-2115){\makebox(0,0)[lb]{\smash{{\SetFigFont{14}{16.8}{\familydefault}{\mddefault}{\updefault}{\color[rgb]{0,0,0}$Q$}%
}}}}
\put(1513,-1311){\makebox(0,0)[lb]{\smash{{\SetFigFont{14}{16.8}{\familydefault}{\mddefault}{\updefault}{\color[rgb]{0,0,0}$v_\ell$}%
}}}}
\put(1986,-1205){\makebox(0,0)[lb]{\smash{{\SetFigFont{14}{16.8}{\familydefault}{\mddefault}{\updefault}{\color[rgb]{0,0,0}$v_r$}%
}}}}
\end{picture}%

%% file: Monge.pstex_t
\begin{picture}(0,0)%
\includegraphics{Monge.pstex}%
\end{picture}%
\setlength{\unitlength}{4144sp}%
\begingroup\makeatletter\ifx\SetFigFont\undefined%
\gdef\SetFigFont#1#2#3#4#5{%
  \reset@font\fontsize{#1}{#2pt}%
  \fontfamily{#3}\fontseries{#4}\fontshape{#5}%
  \selectfont}%
\fi\endgroup%
\begin{picture}(4884,3415)(10214,-3403)
\put(12559,-2798){\makebox(0,0)[lb]{\smash{{\SetFigFont{14}{16.8}{\familydefault}{\mddefault}{\updefault}{\color[rgb]{0,0,0}$\overrightarrow{P}$}%
}}}}
\put(12059,-2797){\makebox(0,0)[lb]{\smash{{\SetFigFont{14}{16.8}{\familydefault}{\mddefault}{\updefault}{\color[rgb]{0,0,0}$\overleftarrow{P}$}%
}}}}
\put(13906,-1816){\makebox(0,0)[lb]{\smash{{\SetFigFont{14}{16.8}{\familydefault}{\mddefault}{\updefault}{\color[rgb]{0,0,0}$R_P$}%
}}}}
\put(14150,-3023){\makebox(0,0)[lb]{\smash{{\SetFigFont{14}{16.8}{\familydefault}{\mddefault}{\updefault}{\color[rgb]{0,0,0}$P_1$}%
}}}}
\put(13204,-2243){\makebox(0,0)[lb]{\smash{{\SetFigFont{14}{16.8}{\familydefault}{\mddefault}{\updefault}{\color[rgb]{0,0,0}$P_2$}%
}}}}
\put(12924,-3090){\makebox(0,0)[lb]{\smash{{\SetFigFont{14}{16.8}{\familydefault}{\mddefault}{\updefault}{\color[rgb]{0,0,0}$u_1'$}%
}}}}
\put(12504,-1835){\makebox(0,0)[lb]{\smash{{\SetFigFont{14}{16.8}{\familydefault}{\mddefault}{\updefault}{\color[rgb]{0,0,0}$v_{|C_2| + 1}'$}%
}}}}
\put(11878,-1836){\makebox(0,0)[lb]{\smash{{\SetFigFont{14}{16.8}{\familydefault}{\mddefault}{\updefault}{\color[rgb]{0,0,0}$v_1'$}%
}}}}
\put(11584,-3042){\makebox(0,0)[lb]{\smash{{\SetFigFont{14}{16.8}{\familydefault}{\mddefault}{\updefault}{\color[rgb]{0,0,0}$u_{|C_1| + 1}'$}%
}}}}
\put(14131,-241){\makebox(0,0)[lb]{\smash{{\SetFigFont{14}{16.8}{\familydefault}{\mddefault}{\updefault}{\color[rgb]{0,0,0}$u_k'$}%
}}}}
\put(12611,-595){\makebox(0,0)[lb]{\smash{{\SetFigFont{14}{16.8}{\familydefault}{\mddefault}{\updefault}{\color[rgb]{0,0,0}$w$}%
}}}}
\put(12947,-1515){\makebox(0,0)[lb]{\smash{{\SetFigFont{14}{16.8}{\familydefault}{\mddefault}{\updefault}{\color[rgb]{0,0,0}$v_{l'}'$}%
}}}}
\put(11758,-1346){\makebox(0,0)[lb]{\smash{{\SetFigFont{14}{16.8}{\familydefault}{\mddefault}{\updefault}{\color[rgb]{0,0,0}$v_{l}'$}%
}}}}
\put(10592,-171){\makebox(0,0)[lb]{\smash{{\SetFigFont{14}{16.8}{\familydefault}{\mddefault}{\updefault}{\color[rgb]{0,0,0}$u_{k'}'$}%
}}}}
\end{picture}%